\documentclass{elsarticle}
\usepackage{amsthm}

\newtheorem{theorem}{Theorem}

\begin{document}

\begin{frontmatter}

\title{Bounds from a Card Trick}
\author{Travis Gagie}
\address{Department of Computer Science\\
    University of Chile}

\begin{abstract}
We describe a new variation of a mathematical card trick, whose analysis leads to new lower bounds for data compression and estimating the entropy of Markov sources.
\end{abstract}


\end{frontmatter}

Several years ago, an article in the popular press~\cite{Ton05} described the following mathematical card trick: the magician gives a deck of cards to an audience member, who cuts the deck, draws six cards and lists their colours; the magician then says which cards were drawn.  The key to the trick is that the magician prearranges the deck so that the sequence of the cards' colours is a substring of a binary De Bruijn cycle of order six, i.e., so that every sextuple of colours occurs at most once.  Although the trick calls only for the magician to name the cards drawn, he or she could also name the next card, for example, with absolute certainty.  At the time we ran across the article, we were studying empirical entropy, and one way to define the $k$th-order empirical entropy of a string $s$ is as our expected uncertainty about the character in a randomly chosen position when given the preceding $k$ characters~\cite{Man01}.  After reading the trick's description, it occurred to us that the $k$th-order empirical entropy of any De Bruijn cycle of order at most $k$ is 0.  Using this and other properties of De Bruijn cycles, we were able to prove several lower bounds for data compression~\cite{Gag09,GM07}.  For example, since $\sigma$-ary De Bruijn cycles of order $k$ have length $\sigma^k$, there are \((\sigma!)^{\sigma^{k - 1}} / \sigma^k\) such sequences~\cite{AEB51} and \(\log_2 \left( (\sigma!)^{\sigma^{k - 1}} / \sigma^k \right) = \Theta (\sigma^k \log \sigma)\), a simple counting argument proves the following theorem.

\begin{theorem}[Gagie, 2006~\cite{Gag06}]
If \(k \geq \log_\sigma n\) then, in the worst case, we cannot store a $\sigma$-ary string $s$ of length $n$ in \(\lambda n H_k (s) + o (n \log \sigma)\) bits for any coefficient $\lambda$.
\end{theorem}

In this paper we consider a variation of the trick described above, that has led us to some new bounds.  This time, suppose the magician does not bother to prearrange the deck, but shuffles it instead and has the audience member draw seven cards; after the audience member lists the cards' colours, the magician has him or her replace the cards, cut the deck again and return it; the magician examines the deck and says which cards were drawn.  It is not hard to show that the probability of two septuples of cards having the same colours in the same order is at most \(1 / 128\) (even if the septuples overlap), so the probability only one sextuple has the colours listed is at least \(1 - 51 / 128 > 0.6\); thus, simply examining the deck gives the magician a better than even chance of guessing the cards drawn.  Our analysis is slightly pessimistic because the probability of two septuples' colours matching would be exactly \(1 / 128\) only if they were drawn with replacement; drawn without replacement, the probability of two cards' colours matching, for example, is \(25 / 51 < 1 / 2\).  Also, even if several sextuples have the colours listed, the magician still has some chance of guessing correctly from amongst them.

Now suppose we draw the $n$ characters of a string $s$ randomly from an alphabet of size $\sigma$.  By the same reasoning as above, the probability two $k$-tuples match is \(1 / \sigma^k\); by linearity of expectation, the expected number of matches is \({n \choose 2} / \sigma^k\).  The 0th-order empirical entropy of $s$
\[H_0 (s)
= (1 / n) \sum_a \mathrm{occ} (a, s) \log_2 (n / \mathrm{occ} (a, s))
\leq \log_2 \sigma\,,\]
where \(\mathrm{occ} (a, s)\) is the number of occurrences of character $a$ in $s$; the $k$th-order empirical entropy of $s$
\[H_k (s)
= (1 / n) \sum_{|\alpha|} |s_\alpha| H_0 (s_\alpha)\,,\]
where \(s_\alpha\) is the concatenation of characters immediately following occurrences in $s$ of the $k$-tuple $\alpha$.  Therefore, calculation shows
\[\mathrm{E} [H_k (s)]
\leq (1 / n) {n \choose 2} \log_2 \sigma
< (n / \sigma^k) \log_2 \sigma\,,\]
which implies the following theorem.

\begin{theorem}
If \(k \geq (1 + \epsilon) \log_\sigma n\) then, in the expected case, we cannot store a $\sigma$-ary string $s$ of length $n$ in \(\lambda n H_k (s) + o (n \log \sigma)\) bits for any coefficient \(\lambda = o (n^\epsilon)\).
\end{theorem}

\begin{proof}
If \(k \geq (1 + \epsilon) \log_\sigma n\) and \(\lambda = o (n^\epsilon)\), then
\[\mathrm{E} \left[ \rule{0ex}{2ex} \lambda n H_k (s) + o (n \log \sigma) \right]
= o (n \log \sigma)\,,\]
but the expected number of bits needed to store $s$ is \(\Theta (n \log \sigma)\).
\end{proof}

\noindent Similarly, by the union bound, the probability there are any matching $k$-tuples at all in $s$ is at most \({n \choose 2} / \sigma^k\), so the probability that \(H_k (s) = 0\) is at least \(1 - {n \choose 2} / \sigma^k\), implying the following theorem.

\begin{theorem}
If \(k \geq (2 + \epsilon) \log_\sigma n\) for some positive constant $\epsilon$ then, with high probability, we cannot store a $\sigma$-ary string $s$ of length $n$ in \(\lambda n H_k (s) + o (n \log \sigma)\) bits for any coefficient $\lambda$.
\end{theorem}

\begin{proof}
If \(k \geq (2 + \epsilon) \log_\sigma n\) for some positive constant $\epsilon$, then
\[\lambda n H_k (s) + o (n \log \sigma)
= o (n \log \sigma)\]
with probability at least \(1 - 1 / n^\epsilon = 1 - o (1)\); however, the number of bits needed to store $s$ is \(\Theta (n \log \sigma)\) with probability \(1 - o (1)\).
\end{proof}

The upper bound above on the probability there are any matching $k$-tuples also quickly yields an exponential lower bound on the sample complexity of estimating the entropy of Markov sources.  This stands in contrast to, e.g., the Shannon-McMillan-Breiman Theorem (see, e.g.,~\cite{CT06}) and bounds for estimating the entropy of a probability distribution~\cite{BDKR05,RRSS09,Val08}.  Although many papers have been written about estimating the entropy of a Markov source (see, e.g.,~\cite{CKV04} and references therein), we know of no previous lower bounds comparable to the one below.

\begin{theorem} \label{thm:estimation}
Suppose a $\sigma$-ary string $s$ is generated either by a deterministic $k$th-order Markov source (which has entropy 0) or by an unbiased memoryless source (which has entropy \(\log_2 \sigma\)).  No algorithm can guess the type of source with probability at least \(2 / 3\) without reading \(\Omega (\sigma^{k / 2})\) characters.
\end{theorem}

\begin{proof}
Suppose there is an algorithm that guesses correctly with probability at least \(2 / 3\) after reading \(o (\sigma^{k / 2})\) characters, when they are generated by an unbiased memoryless source.  By the upper bound above, with high probability the string generated will not contain any matching $k$-tuples.  It follows that we can find a particular string $s$ of length \(o (\sigma^{k / 2})\) containing no matching $k$-tuples and such that, with probability nearly \(2 / 3\), the algorithm classes $s$ as having come from an unbiased memoryless source.  Since $s$ contains no matching $k$-tuples, we can build a deterministic $k$th-order Markov source that generates $s$ with probability 1; on this source, the algorithm errs with probability nearly \(2 / 3\).
\end{proof}

\section*{Acknowledgments}

\noindent Many thanks to Paulina Arena and Elad Verbin for helpful discussions.

\bibliographystyle{elsarticle-num}
\bibliography{trick}

\end{document}